\documentclass[copyright,creativecommons,noderivs]{eptcs}
\providecommand{\event}{GANDALF 2011}

\usepackage{times}
\usepackage{amsmath}
\usepackage{epic,eepic}
\usepackage{amsfonts}
\usepackage{stmaryrd}
\usepackage{paralist}

\def\squarebox#1{\hbox to #1{\hfill\vbox to #1{\vfill}}}
\newcommand{\qed}{\hspace*{\fill}
            \vbox{\hrule\hbox{\vrule\squarebox{.667em}\vrule}\hrule}\smallskip}

\newenvironment{proof}{\begin{trivlist}
\item[\hspace{\labelsep}{\bf\noindent Proof: }]
}{\qed\end{trivlist}}

\newtheorem{THEOREM}{Theorem}[section]
\newenvironment{theorem}{\begin{THEOREM} \hspace{-.85em} {\bf :} }%
                        {\end{THEOREM}}

\newtheorem{rmark}[THEOREM]{Remark}

\newtheorem{xmpl}[THEOREM]{Example}

\newtheorem{observation}[THEOREM]{Observation}

\newtheorem{myclm}[THEOREM]{Claim}
\newenvironment{claim}{\begin{myclm}\rm}{\qed\end{myclm}}

\newcommand{\buchi}{B\"uchi }

\newcommand{\A}{{\cal A}}
\newcommand{\Aphi}{\A_{\varphi}}
\renewcommand{\L}{{\cal L}}

\newcommand{\T}{{\cal T}}

\newcommand{\N}{\mathbb{N}}

\newcommand{\true}{\mbox{\bf true}}

\newcommand{\tuple}[1]{{\langle #1  \rangle}}

\newcommand{\hide}[1]{{}}

\newcounter{sidecommentcounter}

\newcommand{\Path}{\mbox{\it Path}}
\newcommand{\call}{\mbox{\it call}}
\newcommand{\ret}{\mbox{\it ret}}
\newcommand{\nw}{\overline}
\newcommand{\next}{\bigcirc}
\newcommand{\prev}{\mbox{$\bigcirc \hspace{-9pt} - \hspace{2pt}$} }
\newcommand{\munext}{\next_\mu}
\newcommand{\muprev}{\prev_\mu}
\newcommand{\until}{\mbox{\bf U}}
\newcommand{\since}{\mbox{\bf S}}
\newcommand{\sumuntil}{\mbox{$\until^\sigma$}}
\newcommand{\sumsince}{\mbox{$\since^\sigma$}}
\newcommand{\re}{ s_0 }
\newcommand{\dir}{ s_C }
\newcommand{\rhopath}{\mbox{$\rho$-$\Path$}}

\title{Synthesis from Recursive-Components Libraries%
\thanks{For a longer version of this paper see
{\tt http://www.cs.rice.edu/$\sim$vardi/papers}.}}
\author{
Yoad Lustig\thanks{Current address: 
Yahoo! Labs Haifa,
Matam Scientific Industries Center Building \#3,
Matam Park, Haifa, 31905 Israel,
email: yoad@yahoo-inc.com}
\institute{Rice University\\
6100 Main Street\\
Houston, TX 77005-1892, USA}
\email{yoad.lustig@gmail.com}
\and
Moshe Y. Vardi\thanks{
Work supported in part by NSF grants 
CCF-0728882, and CNS 1049862, by BSF grant 9800096, 
and by gift from Intel.}
\institute{Rice University\\
6100 Main Street\\
Houston, TX 77005-1892, USA}
\email{vardi@cs.rice.edu}
}
\def\titlerunning{Synthesis from Recursive-Components Libraries}
\def\authorrunning{Yoad Lustig and Moshe Y. Vardi}
\def\firstpage{1}
\begin{document}
\maketitle

\begin{abstract}
Synthesis is the automatic construction of a system from its
specification. In classical synthesis algorithms it is always
assumed that the system is "constructed from scratch" rather than
composed from reusable components. This, of course, rarely happens
in real life. In real life, almost every non-trivial commercial
software system relies heavily on using libraries of reusable
components. Furthermore, other contexts, such as web-service 
orchestration, can be modeled as synthesis of a system from a 
library of components. 

In 2009 we introduced LTL synthesis from libraries of reusable components. 
Here, we extend the work and study synthesis from component libraries
with ``call and return'' control flow structure.  Such control-flow 
structure is very common in software systems.  We define the problem of 
Nested-Words Temporal Logic (NWTL) synthesis from recursive component 
libraries, where NWTL is a specification formalism, richer than LTL, that 
is suitable for ``call and return'' computations.  We solve the problem, 
providing a synthesis algorithm, and show the problem is 
2EXPTIME-complete, as standard synthesis.
\end{abstract}

\section{Introduction}
The design of almost every non-trivial software system is based on
using libraries of reusable components. Reusable components come in
many forms: functions, objects, or others. Nevertheless, the basic idea 
of constructing systems from reusable components underlies almost all 
software construction.  Indeed, almost every system involves many 
sub-systems, each dealing with different engineering aspects and each 
requiring different expertise. In practice, the developer of a commercial 
product rarely develops all the required sub-systems herself. 
For example, a software application for an email client contains
sub-systems for managing graphic user interface (as well as many other 
sub-systems).  Rarely will a developer of the email-client system develop 
the basic graphic-user-interface functionality as part of the project.  
Instead, basic sub-systems functionality is usually acquired as a  
\emph{library}, i.e., a collection of reusable components that can be 
integrated into the system.  The construction of systems from reusable
components is extensively studied.  Many examples for important work on 
the subject can be found in Sifakis' work on component-based construction 
\cite{Sif05} and  de Alfaro and Henzinger's work on ``interface-based 
design'' \cite{dAH05}.  Furthermore, other situations, such as web-service
orchestration \cite{BCGLM03,SPG07}, can be viewed as the
construction of systems from libraries of reusable components.
  
Synthesis is the automated construction of a system from its specification.
The basic idea is simple and appealing: instead of developing a system and
verifying that it adheres to its specification, we would like to have an 
automated procedure that, given a specification, constructs a system that 
is correct by construction.  
The modern approach to temporal synthesis was 
initiated by Pnueli and Rosner, who introduced LTL (linear temporal logic) 
synthesis \cite{PR89a}.  In LTL synthesis, the specification is given in 
LTL and the system constructed is a finite-state transducer modeling a 
reactive system.  In this setting of synthesis it is always assumed that 
the system is ``constructed from scratch'' rather than ``composed'' from 
reusable components.  In \cite{LV09}, we introduced the study of synthesis
from reusable components. We argued there that even when it is 
theoretically possible to design a sub-system from scratch, it is often 
desirable to use reusable components.  The use of reusable components 
allows abstracting away most of the detailed behavior of the sub-system, 
and writing a specification that mentions only the aspects of the 
sub-system relevant for the synthesis of the system at large.

A major concern in the study of synthesis from reusable components is the 
choice of a mathematical model for the components and their composition.
The exact nature of the reusable components in a software library may
differ. The literature, as well as the industry, suggest many different 
types of components; for example, function libraries (for procedural
programming languages) or object libraries (for object-oriented
programming languages).  Indeed, there is no one correct model 
encompassing all possible facets of the problem. The problem of synthesis 
from reusable components is a general problem to which there are as
many facets as there are models for components and types of composition.  
Components can be composed in many ways: synchronously or asynchronously, 
using different types of communications, and the like \cite{Sif05}.

As a basic model for a component, following \cite{LV09}, we abstract away 
the precise details of the component and model a component as a 
{\em transducer}, i.e., a finite-state machine with
outputs. Transducers constitute a canonical model for reactive
components, abstracting away internal architecture and focusing on
modeling input/output behavior.  In \cite{LV09}, two models of 
composition were studied. In \emph{data-flow} composition
the output of one component is fed as input to another component. 
The synthesis problem for data-flow composition was shown to be 
undecidable.  In \emph{control-flow} composition
control is held by a single component at every point in time; the
composition of components amounts to deciding how control is
passed between components, by setting which component receives control
when another component relinquishes it. Control-flow is motivated by
software (and web services) in which a single function is in control
at every point during the execution. In \cite{LV09} we focused on 
``goto'' control flow, and proved that LTL synthesis in that setting is 
2EXPTIME-complete.

In this paper we extend that work and study a composition
notion that relates to ``call and return'' control structure. 
``Call and return'' control flow is very natural for both software and
web services.  An online store, for example, may ``call'' the PayPal web
service, which receives control of the interaction with the user until
it returns the control to the online store.  To allow for 
``call and return'' control-flow structure, we define  
a recursive component to be a transducer in which some of the
states are designated as exit states.  The exist states are partitioned 
into call states, and return states.  Intuitively, a recursive component 
receives control when entering its initial state and relinquishes control 
when entering an exit state.  When a call state is entered, the control is
transferred from the component in control to the component that is
being called by the component in control. When a return state is entered, 
the control is transferred from the component in control to the component 
that called it (i.e., control is returned).  To model return values, 
each transducer has several return states. Each return state is associated
with a re-entry state.  Thus, each transducer has a single entry state, 
several re-entry states, several return states, and several call states.
Composing recursive components amounts to matching call states with
entry states and return states with re-entry states.\footnote{
It is possible to consider more complex models, for example, models in
which there are several call values. The techniques presented here
can be extended to deal with such models. 
} 

Dealing with ``call and return'' control flow poses two distinct
conceptual difficulties. The first is the technical difficulty of
dealing with a ``call and return'' system that has a pushdown store. 
When adapting the techniques of \cite{LV09}, a run is no longer a path
in a control-flow tree, but rather a traversal in a composition tree, in 
which a return corresponds to climbing up the tree. To deal with this 
difficulty we employ techniques used with 2-way automata \cite{PV01}.  
A second difficulty has to do with the specification language.
``Call and return'' control-flow requires a richer specification
language than LTL \cite{AEM04,AABEIL08}.  For example, one might like to 
specify that one function is only called when another function is in the 
caller's stack; or that some property holds for the local computations of 
some function.  In recent years an elegant theory of these issues was 
developed, encompassing suitable specification formalisms, as well as 
semantic, automata-theoretic, and algorithmic issues 
\cite{AEM04,AABEIL08,AM09}.  Here we use the specification language  
\emph{nested-words temporal logic} (NWTL) \cite{AABEIL08}, 
and the automata-theoretic tool of \emph{nested words \buchi automata} 
(NWBA) \cite{AABEIL08,AM09}.

We define here and study the NWTL recursive-library-component
realizability and synthesis problems.  We show that the complexity of
the problem is 2EXPTIME-complete (like standard synthesis and
synthesis of ``goto'' components) and provide a 2EXPTIME algorithm for
the problem. We use the composition-tree technique of \cite{LV09}, in
which a composition is described as an infinite tree. The challenge
here is that we need to find nested words in classical trees. While
the connection between nested words and trees has been studied
elsewhere, cf.~\cite{Alu07}, our work here is the first to combine
nested-word automata with the classical tree-automata framework for
temporal synthesis, using techniques developed for two-way automata
\cite{PV01,Var98}. 

\section{Preliminaries}

\noindent
{\sf Transducers}:
A {\em transducer}
is a deterministic automaton with outputs;
$\T = \tuple{ \Sigma_I, \Sigma_O, Q, q_0, \delta, F, L}$, 
where: $\Sigma_I$ is a finite input alphabet, 
$\Sigma_O$ is a finite output alphabet,
$Q$ is a set of states,
$q_0\in Q$ is an initial state,
$\delta:Q\times\Sigma_I \to Q$ is a transition function, 
$F$ is a set of final states,
and $L:Q\to\Sigma_O$ is an output function labeling states with
output letters.  For a transducer $\T$ and an input word 
$w = w_1 w_2 \ldots w_n \in \Sigma^n_I$, a {\em run}, or a 
{\em computation} of $\T$ on $w$ is a sequence of states 
$r = r_0, r_1, \ldots r_n \in Q^n$ such that  
$r_0 = q_0$ and for every $i\in [n]$ we have 
$r_i = \delta(r_{i-1},w_i)$. 

For a transducer $\T$, we define $\delta^*: \Sigma_I^* \rightarrow Q$ in
the following way: $\delta^*(\varepsilon)= q_0$,
and for $w\in \Sigma_I^*$ and $\sigma \in \Sigma_I$, we have 
$\delta^*(w \cdot \sigma)=\delta(\delta^*(w),\sigma)$.
A $\Sigma_O$-labeled $\Sigma_I$-tree $\tuple{\Sigma_I^*,\tau}$ is 
{\em regular} if there exists a transducer 
${\cal T}=\tuple{\Sigma_I,\Sigma_O,Q, q_0,\delta,L}$
such that for every $w\in \Sigma_I^*$, we have $\tau(w)=L(\delta^*(w))$.
A transducer $\T$ outputs a letter for every input letter it
reads. Therefore, for an input word $w_I\in\Sigma_I^\infty$,
the transducer $\T$ induces a word 
$w\in (\Sigma_I \times \Sigma_O)^\infty$ 
that combines the input and output of $\T$. The 
\emph{maximal computations} of $\T$ are those that exit at a final 
state in $F$ or are of length~$\omega$.

\noindent
{\sf Nested Words, NWTL and NWBA}:
When considering a run in the ``call and return'' control-flow model, 
the run structure should reflect both the linear order of
the execution and the matching between calls and their
corresponding returns. For example, when a programmer uses a debugger
to simulate a run, and the next command to be executed is a call, 
there are two natural meanings to ``simulate next command'':
first, it is possible to execute the next machine command to be
executed (i.e. jump into the called procedure). In debugger
terminology this is ``step into'', and this meaning reflects the
linear order of machine commands being executed. On the other hand,   
it is possible to simulate the entire computation of the procedure
being called, i.e. every machine command from the call to its
corresponding return. In compiler terminology this is ``step over'',
and this meaning reflects the matching between calls and their returns.
Thus, the structure of a run, with the matching between calls and
returns, is richer then the sequence of commands that reflects only
the linear order.
Relating to this richer structure is crucial for reasoning about
recursive systems, and it should be reflected in the mathematical
model of a run, in the formalism by which formal claims on runs are
made, i.e., in the specification formalism.

A run in a ``call and return'' model is a sequence
of configurations, or {\em a word}, together with a matching relation that
matches calls and their corresponding returns. 
The matching relation is {\em nested}, i.e. constrained to ensure that
a return to an inner call appears before the return to an outer call. 
A formal definition appears below. The model of the run consists of
both the word (encoding the linear order) and the matching relation.
A word with nested matching is a {\em nested word} \cite{AM09}.
At the specification level, it should be possible to make formal claims 
regarding system that refer to the ``call and return'' structure
\cite{AEM04,AABEIL08}.
For example: one may want to argue about the value of some memory
location as long as a function is in scope (i.e. during the subsequence
of the computation between the call to the function and its
corresponding return).  Alternatively one may
want to argue about the values of some local values whenever some
function is in control (that may correspond to several continuous
subsequences of commands).  
Another example is arguing about the call stack whenever some function
is in control (such as ``whenever f is in control either g or h are on
the call stack'').  
Several specification formalisms were suggested to reason about ``call
and return'' computations \cite{AEM04,AABEIL08,AM09}. Here we use
\emph{Nested Words Temporal Logic},
(NWTL) \cite{AABEIL08}, which is both expressive and natural to use. 
Finally, to reason about nested words, we use \emph{nested words \buchi 
automata} (NWBA), which are a special type of automata that run on
nested words  \cite{AABEIL08,AM09}.  Intuitively, in a standard infinite 
word, each letter has a single successor
letter.  Therefore, automata on standard words can be seen as being in some
state $q$, reading a letter $\sigma$ and ``sending'' the next state $q'$
to the successor letter $\sigma'$. 
In a nested word, however, a letter $\sigma$ might have two ``natural
successors''. First the letter $\sigma'$ following it in the linear 
sequence of execution, and second another letter $\sigma''$ that is 
matched to it by the ``call and return'' matching. 
A NWBA not only ``sends''
a state to the successor letter $\sigma$, but also ``sends'' some
information, named {\em hierarchical symbol}, to the
matched letter $\sigma''$. The transition relation takes into account
both the state and the hierarchical symbols. A formal definition of 
NWBA's is presented below.

We proceed with the formal definitions of nested words, the logic NWTL
for nested words, and the automata NWBA running on nested words.
The material presented below is taken from \cite{AABEIL08}, which we
recommend for a reader who is not familiar with nested words, their
logic, or their automata.

A {\em matching} on $\N$ or an interval $[1, n]$ of $\N$ is a binary
relation  $\mu$ and two unary relations $\call$ and $\ret$, satisfying
the following: 
(1) if $\mu(i, j)$ holds then $\call(i)$ and $\ret(j)$ hold and $i < j$; 
(2) if $\mu(i, j)$ and $\mu(i, j')$ hold then $j = j'$ and if 
    $\mu(i, j)$ and $\mu(i', j)$ hold then $i = i'$;  
(3) if $i \le j$ and $\call(i)$ and $\ret(j)$ hold, then there exists 
    $i \le k \le j$ such that either $\mu(i, k)$ or $\mu(k, j)$.
Let $\Sigma$ be a finite alphabet. A finite nested word of length $n$
over $\Sigma$ is a tuple $\nw{w} = \tuple{w, \mu, \call, \ret}$, where 
$w = a_1 \ldots a_n \in \Sigma^*$, and $\tuple{\mu, \call, \ret}$ is a
matching on $[1, n]$. 
A nested $\omega$-word is a tuple $\nw{w} = \tuple{w, \mu, \call, \ret}$,
where $w = a_1 \ldots \in \Sigma^\omega$, and $\tuple{\mu, \call, \ret}$ 
is a matching on $\N$.
We say that a position $i$ in a nested word $\nw{w}$ is a call position
if $\call(i)$ holds; a return position if $\ret(i)$ holds; and an
internal position if it is neither a call nor a return. 
If $\mu(i, j)$ holds, we say that $i$ is the matching call of $j$, and
$j$ is the matching return of $i$, and write $c(j) = i$ and $r(i) = j$. 
Calls without matching returns are pending calls.
For a nested word $\nw{w}$, and two positions $i, j$ of $\nw{w}$, we
denote by $\nw{w}[i, j]$ the substructure of $\nw{w}$ (i.e., a finite
nested word) induced by positions $l$ such that $i \le l \le j$. 
If $j < i$ we assume that $\nw{w}[i, j]$ is the empty nested word. 
For nested $\omega$-words $\nw{w}$, we let $\nw{w}[i,\infty]$ denote 
the substructure induced by positions $l \ge i$.
When this is clear from the context, we do not distinguish references
to positions in subwords $\nw{w}[i, j]$ and $\nw{w}$ itself, e.g., we
shall often write $ \tuple{\nw{w}[i, j], i} \models \varphi$ to mean
that $\varphi$ is true at the first position of $\nw{w}[i, j]$.


{\em Nested words temporal logic} ({\em NWTL}) is a
specification formalism suitable for ``call and return'' computations
\cite{AABEIL08}.
First we define a summary path between positions $i<j$ in a
nested word $\nw{w}$. Intuitively, a summary path skips from calls to
returns on the way from $i$ to $j$. The {\em summary path} between
positions $i<j$ in a nested word $\nw{w}$ is a sequence 
$i = i_0 < i_1 < \ldots < i_k = j$ such that for all $p < k$ we have 
$i_{p+1} = r(i_p)$ if $i_p$ is a matched call and $j
\ge r(i_p)$; or $i_{p+1} = i_p +1$ otherwise.
Next, we define NWTL syntax.
For an alphabet $\Sigma$, the letters of $\Sigma$, $\top$ (standing
for true), $\call$, and $\ret$ are NWTL formulas.  
NWTL has the operators: not $\lnot$, or $\lor$, next $\next$, abstract
next (that skips from a call to its return) $\munext$, previous $\prev$,   
abstract previous $\muprev$, summary until (to be defined below) 
$\sumuntil$, and summary since $\sumsince$. 
For NWTL formulas $\varphi_1, \varphi_2$ the following are NWTL formulas:
$\lnot \varphi_1 |
\varphi_1 \lor \varphi_2 | \next \varphi_1 | \munext \varphi_1 | 
\prev \varphi_1 | \muprev \varphi_1 | \varphi_1 \sumuntil \varphi_2 | 
\varphi_1 \sumsince \varphi_2$.
We proceed to define NWTL semantics.
Let $w = w_1 \ldots w_n$ or $w_1\dots$ be a finite or infinite word
over $\Sigma$.  
Let $\nw{w} = \tuple{ w, \call, \ret, \mu}$, and $i\ge 1$ be a number
bounded by the length of $w$.
Every nested word satisfies $\top$, in particular $(\nw{w},i)\models \top$.
For a letter $\sigma\in \Sigma$ we have $(\nw{w},i)\models \sigma$ iff 
$\sigma = w_i$. (This is can be extended to alphabets of the type  
$\Sigma = 2^{AP}$,  that consists of
sets of atomic propositions, in the standard way, i.e.,
$(\nw{w},i)\models p$ iff $p\in w_i$). 
Boolean operators semantics is standard $(\nw{w},i)\models \lnot \varphi$
iff $(\nw{w},i)\not \models \varphi$; and 
$(\nw{w},i)\models \varphi_1 \lor \varphi_2$ iff 
$(\nw{w},i)\models \varphi_1$ or $(\nw{w},i)\models  \varphi_2$.
We also have $(\nw{w},i)\models \next \varphi$ iff 
$(\nw{w},i+1)\models \varphi$ and $(\nw{w},i)\models \prev \varphi$ iff 
$(\nw{w},i-1)\models \varphi$. We have $(\nw{w},i) \models \call$ iff
$i$ is a call, and $(\nw{w},i) \models \ret$ iff $i$ is a return.
We have $(\nw{w},i)\models \munext \varphi$ iff $i$ is a call with a
matching return $j$ (i.e., $\mu(i, j)$ holds) and 
$(\nw{w}, j) \models \varphi$. Similarly, $(\nw{w},i)\models \muprev
\varphi$ iff $i$ is a return with a matching call $j$ (i.e.,
$\mu(j,i)$ holds) and $(\nw{w}, j) \models \varphi$.
For summary until we have 
$(\nw{w},i)\models \varphi_1 \sumuntil \varphi_2$ iff there exists a
$j\ge i$ for which $(\nw{w},j)\models \varphi_2$, and for the summary
path $i = i_0 < i_1 < \ldots < i_k = j$ between $i$ and $j$ we have 
for every $p< k$ that $(\nw{w},i_p)\models \varphi_1$. Similarly, 
$(\nw{w},i)\models \varphi_1 \sumsince \varphi_2$ iff there exists a
position $j < i$ for which $(\nw{w},j)\models \varphi_2$ and for the
summary path  $j = i_0 < i_1 < \ldots < i_k = i$ between $j$ and $i$
we have for every $p\in [k]$ that $(\nw{w},i_p)\models \varphi_1$. 

Rather than use NWTL directly, we use here \emph{nested-word
\buchi automata} (NWBA), which are known to be at least
as expressive as NWTL; in fact, there is an exponential translation from
NWTL to NWBA \cite{AABEIL08}, analogous to the exponential
translation of linear temporal logic to \buchi automata \cite{VW94}.
A {\em nondeterministic nested word \buchi automaton} ({\em NWBA})
is a tuple
$\A=\tuple{\Sigma,Q,Q_0,Q_f,P,P_0,P_f,\delta_c,\delta_i,\delta_r}$,
consisting of a finite alphabet $\Sigma$,
finite set $Q$ of states, a set $Q_0\subseteq Q$ of
initial states, a set $Q_f\subseteq Q$ of accepting states,
a finite set $P$ of hierarchical symbols, a set $P_0\subseteq P$ of initial
hierarchical symbols, a set $P_f\subseteq P$ of final hierarchical
symbols, a call-transition relation  
$\delta_c \subseteq Q \times \Sigma \times Q \times P$, an internal
transition relation $\delta_i \subseteq Q \times \Sigma \times Q$, and
a return-transition relation 
$\delta_r \subseteq Q \times P \times \Sigma \times Q$. 
The automaton $\A$ starts in an initial state and reads the nested word
from left to right. 
A run $r$ of the automaton $\A$ over a nested word 
$\nw{w} = \tuple{ a_1 a_2 \ldots, \mu, \call, \ret}$ 
is a sequence $q_0, q_1, \ldots$ of states, 
and a sequence $p_{i_1}, p_{i_2}, \ldots$ of hierarchical symbols, 
corresponding to the call positions $i_1, i_2, \ldots$,
such that $q_0\in Q_0$, and for each position $i$, if $i$ is a call then 
$\tuple{q_{i - 1}, a_i, q_i, p_i} \in \delta_c$; 
if $i$ is internal, then $\tuple{q_{i - 1}, a_i, q_i} \in \delta_i$; 
if $i$ is a return such that $\mu(j, i)$, then 
$\tuple{q_{i - 1}, p_j , a_i, q_i} \in \delta_r$; 
and if $i$ is an unmatched return then
$\tuple{q_{i - 1}, p, a_i, q_i} \in \delta_r$ for some $p\in P_0$. 
Intuitively, in a run~$r$, the hierarchical symbol
associated with a matched return position~$i$, is the hierarchical
symbol $p_j$, associated with the call position~$j$ that is matched to~$i$.
The run $r$ is accepting if 
(1) for all pending calls $i$, $p_i\in P_f$ , and 
(2) if $\nw{w}$ is a finite word of length $l$ then the final state
$q_l$ is accepting (i.e.,  $q_l \in Q_f$),
and if $\nw{w}$ is an $\omega$-word then for infinitely many
positions $i$, we have $q_i\in Q_f$. 
The automaton $\A$ accepts the nested word $\nw{w}$ if it has an
accepting run over $\nw{w}$.

\section{The computational model}\label{sec: RLCs}

\noindent
{\sf Recursive Components and their composition}:
To reason about recursive components one has to choose a mathematical
model for components. The choice of model has to balance the need for
a rich modeling formalism, for which computationally powerful
models are preferred, and the need to avoid the pitfall of
undecidability, for which simpler models are preferred. 

A successful sweet spot in this trade off is the computational model
of finite-state transducers, i.e. finite-state machines with output.  
A common approach to reasoning about real world systems, is
abstracting away the data-intensive aspects of the computation and
model the control aspects of the computation by a finite-state
transducer. Using this approach, the transducers model is rich enough
to model real world industrial designs \cite{HOL97,BCLR04}.
For that reason, transducers are widely used in both theory 
\cite{VW94,PR89a,ABEGRY05} and practice \cite{HOL97,BCLR04}, 
and are prime candidates as a model for ``call and return'' components.

To model ``call and return'' control-flow by transducers, 
we introduce a small variation on the basic transducer model. 
Essentially, we use transducers in which some states are ``call
states'', where a transition to one of these states stands for a call to
another component; some states are ``return'' states, where a
transition to one of these states stands for a return to 
the component that called this component; and some states are
re-entry states, i.e., states to which the component enters upon
return from a call to another component. Similar models can be found in 
\cite{ABEGRY05}.
Different return values, are modeled here by having different
re-entry states. The model is somewhat simplified in the sense that a 
return is not constrained in terms of the call state through which the 
call was made. In software, for example, the return is constrained to the 
instruction following the call instruction (although several return values
may be permitted). Nevertheless, the model is rich enough to deal with
the essence of ``calls and returns'', and the techniques we present can be
used to deal with richer models (e.g. each call may be associated with a
mapping between return states and re-entry states capturing constrained 
returns as above). We chose this simpler model as it allows for simpler 
notation and clearer presentation of the underlying ideas.   

To simplify the notation, we fix a number $n_C$ and assume
every component in the library has exactly $n_C$ calls. Similarly, 
we fix a number $n_R$ and assume every component in the library has
exactly $n_R$ return points, as well as exactly $n_R$ points to which
the control is passed upon return.

A {\em Recursive Library Component} (\emph{RLC}) is a finite
transducer with call, return and re-entry states.
Formally, an RLC is a tuple 
$M =$ $\langle \Sigma_I,$ $\Sigma_O,$ $S,$ $s_0,$ $s_e^R,$ $S_C,$ $S_R,$ $\delta, L\rangle$
where:
(1) $\Sigma_I$ and $\Sigma_O$ are finite input and output alphabets.
(2) $S$ is a finite set of states.
(3) $s_0\in S$ is an initial state. When called by another component,
the component $M$ enters $s_0$.
(4) $s_e^R\subseteq S$ is a set of re-entry states.
When the control returns from a call to another component, $M$
enters one of the re-entry states in $s_e^R$. 
We denote $s_e^R = \{ s_e^1, \ldots, s_e^{n_R} \}$  
(5) $S_C\subseteq S$ is a set of call states. When $M$ enters a state in 
$S_C$, another component $M'$ is called, and the control is transferred to
$M'$ until control is returned. 
We denote $S_C = \{ s_C^1, \ldots, s_C^{n_C} \}$  
(6) $S_R\subseteq S$ is a set of return states. 
When $M$ enters a return state, the control is returned to the
component that called $M$. 
We denote $S_R = \{ s_R^1, \ldots, s_R^{n_R} \}$.
When the $i$-th return state, i.e. $s_R^i$, is entered, control is
returned to the caller component $M'$, which is entered at his $i$-th
re-entry state (i.e., $M'$'s state $s_e^i$).   
(7) $\delta: S\times \Sigma_I \to S$ is a transition function.
(8)  $L:S\to \Sigma_O$ is an output function, labeling each state by an
output symbol.

The setting we consider is the one in which we are given a library
$\L = \{ C_1, \hide{C_2,} \ldots, C_l \}$ of RLC components.
A \emph{composition} over $\L$ is a tuple 
$\tuple{(1, C_1, f_1),(2, C_2, f_2),\ldots,(k, C_k, f_k)}$ 
of~$k$ composition elements, in which each composition element is a triple
composed of an index $i$, an RLC $C_i\in\L$, and an interface function
$f_i:S_C \to [k]$ that maps each of $C_i$'s call states into the 
composition element that is called upon entry to the call state.  
Note that the same RLC can be instantiated in different elements of the
composition, but with different interface functions, and the size
of the composition is a priori unbounded.%
\footnote{
If we had bounded the number of elements in a composition, then
the number of ways in which these elements can be composed would have been
finite and the search for a composition that satisfies some
specification would have turned into a combinatorial search,
analogously, for example, to \emph{bounded synthesis} \cite{KLVY11}.
} 
While we consider here only finite compositions, we could have
considered, in principle, also infinite compositions. As we shall
see, for NWBA specifications, finite compositions are sufficient.

A run of the system begins in state $s_0$ of $C_1$.
When the run is in a state of the component $C$ we say that the
component $C$ is in control. For example, a run begins when the
component $C_1$ is in control.  
For every $i\le k$, as long as a component $C_i$ is in control, the 
system behaves as $C_i$ until an exit state (i.e. a
call state or a return state) is entered. 
If a call state $s_C^j\in S_C$ of $C_i$ is entered then the component
$C_{f_i(j)}$ is called. That is, the control is passed to the
$f_i(j)$-th component in the composition. The run proceeds
from the start state of $C_{f_i(j)}$. 
If a return state $s_R^j\in S_R$ of $C_i$ is entered (when
$C_i$ is in control), then $C_i$ returns the control to the
component that called $C_i$. If, for example, $C_i$ was called by
$C_j$ then when $s_R^m$ is entered, the run proceeds from the
re-entry state $s_e^m$ of $C_j$. 
We now define the composition formally.

Formally, a composition
$C=\langle (1, C_1, f_1),(2, C_2, f_2),\ldots,(k, C_k, f_k) \rangle$,
where\\\
$C_i = \langle \Sigma_I, \Sigma_O, S[i], s_0[i], s_e^R[i],
S_C[i], S_R[i], \delta[i], L[i] \rangle$,
induces a (possibly infinite) transducer \\
$M = \{\Sigma_I, S_O, s_0^M, \delta^M, L^M \}$, where:
\begin{compactenum}
\item
  The input alphabet is $\Sigma_I$ and the output alphabet is $\Sigma_O$.
\item
  The states of $M$ are 
  finite sequences of the form $\tuple{ i_1, i_2, \ldots, i_m, s }$, where
  for every $j\le m$ we have $i_j\in [k]$, and the final
  element is a state $s\in S[i_m]$ of $C_{i_m}$. Intuitively, such a
  state stands 
  for the computation being in the state $s$ of the RLC $C_{i_m}$,
  where the computation call stack is $i_1, i_2, \ldots, i_m$.
  The initial state of $M$ is $\tuple{1,s_0[1]}$ where $s_0[1]$ is the
  initial state of $C_1$. 
  Formally, $S_M  = [k]^* \cdot (\bigcup_{i\in [k]} i\cdot S[i])$.

\item
  Next, we define the transition function $\delta^M$. 
  Let $v = \tuple{ i_1, i_2,\ldots , i_m , s }$ be a state of $M$.
  Then, $\delta^M(v, \sigma) = v'$ if one of the following holds:
  \begin{compactenum}
  \item
    {\bf internal transition}:
    If $\delta[i_m](s, \sigma) = s'$ for some state 
    $s'\in S[i_m]\setminus (S_C[i_m]\cup S_R[i_m])$ of $C_{i_m}$, then
    $v' = \tuple{ i_1, \ldots , i_r , s' }$, where
  \item
    {\bf call transition}:
    If $ \delta[i_m](s, \sigma) = s'$ where $s'\in S_C[i_m]$ is  
    the $j$-th call state of $C_{i_m}$ (i.e., $s' = s_C^j[i_m]$), then
    $v' = \tuple{ i_1, \ldots , i_m , f_{i_m}(j), s_0[f_{i_m}(j)]}$,
  \item
    {\bf return transition}:
    If $\delta[i_m](s, \sigma) = s'$ where $s' \in S_R[i_m]$, 
    is the $j$-th return state of $C_{i_m}$ (i.e., $s' = s_R^j[i_m]$),
    then $v' = \tuple{ i_1, \ldots , i_{m-1} , s_e^j[i_{m-1}] }$.
  \end{compactenum}
\item
The final state set $F^M=\langle 1,S_R[1]\rangle$. Intuitively, the
computation terminates when the first component returns.
\item
The output function $L^m$ is defined by
$L^m(\tuple{i_1, \ldots, i_m,s}) = L[i_m](s)$.  
\end{compactenum}

For an input word 
$w^I = w^I_0,w^I_1\ldots \in \Sigma_I^\infty$, 
the transducer $M$ induces an output word  
$w^O = w^O_0, w^O_1, \ldots \in \Sigma_O^\infty$. 
We denote by  $w = (w^I_0, w^O_0),(w^I_1, w^O_1)\ldots$
the combined input-output sequence induced by $w^I$.
Furthermore, on the input word $w^I$, the composition $C$ induces a
nested word $\nw{w} = \tuple{ w, \call, \ret, \mu}$ in which $w$ is
the input-output induced word, $\call$ holds in positions in which a
component made a call, $\ret$ holds in positions in which a component
returned, and $\mu$ maps each call to its return.
We sometime abuse notation and refer to the word $w$ rather than the
nested word $\nw{w}$. Similarly we might refer to a {\em computation}
of, or in, a composition meaning a nested word induced by the
composition. Similarly, we may refer to a {\em computation segment}
meaning a substructure $\nw{w}[i,j]$, for some positions $i,j$, of a
computation.

A composition $C$ {\em realize} an NWTL specification $\varphi$ if
all computations induced by $C$ satisfy $\varphi$.
The {\em recursive-library-components realizability problem} is:
given a library of RLCs $\L = \{ M_j\}_{j=1}^n$ and an NWTL
specification $\varphi$, decide whether there exists a composition of  
components from the library that realize $\varphi$.
The {\em recursive-library-components-synthesis problem} is: given a
library of RLCs $\L = \{ M_j\}_{j=1}^n$ and an NWTL specification
$\varphi$, decide whether $\varphi$ is realizable by a composition of
RLCs from $\L$ and if so, output a composition realizing $\varphi$. 

\noindent
{\sf Composition trees}\label{sec: Composition trees}
Next, we define the notion of a {\em composition tree}, which is the
analog of a control-flow tree in \cite{LV09}.
Fixing a library $\L$ of RLCs, composition trees represent compositions. 
A composition tree is labeled tree $\tau = \tuple{T, \lambda}$, where
$T$, the tree structure, is the set $[n_C]^*$, and $\lambda:T\to \L$ 
is a mapping of the tree vertexes into $\L$.  Every composition 
$C  = \tuple{(1, C_1, f_1),(2, C_2, f_2),\ldots,(k,  C_k, f_k)}$, 
induces an $\L$-labeled composition tree $\tau_C$.
We first show  that $C$ induces a $[k]$-labeled tree that we call 
{\em intermediate tree}.   A labeled tree $\tuple{[n_C]^*,\kappa}$,
where $\kappa:[n_C]^*\to [k]$, is the intermediate mapping induced by
$C$, if $\kappa(\epsilon) =1$, and, for every $v\in [n_C]^*$
and $j\in [n_C]$, we have that $\kappa(v\cdot j)=f_{\kappa(v)}(j)$. 
The {\em composition tree} induced by $C$ is $\tuple{[n_C]^*,\lambda}$
where for every $v\in [n_C]^*$ we have that $\lambda(v)=C_{\kappa(v)}$.
A node $v = i_1 \cdots i_k$ represents a call-stack
configuration. The node's label $\lambda(v)$ is the component in
control, while the labels of the node's successors, i.e.,
$\lambda(v\cdot 1), \ldots, \lambda(v\cdot n_C)$, stand for the 
components that are called if a call state is entered.
Intuitively, the control flow of an actual computation is a represented
by a traversal in a composition tree. The control is first given to
the component labeled by the root. For a node~$v$, a call corresponds
to a descent to a successor (where a call from the $i$-th call state
corresponds to a descent to the $i$-th successor). Similarly, a return 
from a node $v$ corresponds to an ascent to the predecessor of~$v$. 

Thus, a composition induces a composition tree. On the other hand, a
composition tree can be seen as an ``infinite composition'' in which
each node $v$ stands for a composition element in which the component
is the label of $v$, and the interface function $f_v$ maps the call
states to the successors (i.e., for every $v\in [n_C]^*$ and 
$i\in [n_C]$ we have $f_v(i) = v\cdot i$). 
So a composition tree induces an infinite composition. 
We abuse terminology and refer to computations of a composition tree,
where we mean to refer to computations of the induced infinite composition.
Furthermore, in Theorem~\ref{thm: RLC synthesis solution} we show 
how a finite composition can be extracted from a \emph{regular} 
composition tree.  Another abuse of terminology we make is to refer to a 
labeled subtree of a composition tree as a composition tree. 

\section{Recursive-library-components synthesis algorithm}
Our approach to the solution of the RLC synthesis problem, is first to
construct a tree-automaton $\A_b$ that accepts composition trees that
do {\em not} satisfy the specification. Once that is 
achieved, $\A_b$ can be complemented to get an automaton $\A$ which is 
accepts composition trees that \emph{do} satisfy the specification. 
Finally, $\A$'s language can be checked for emptiness
and if not empty, a system can be extracted from a witness (similar to
the algorithm in \cite{LV09}).
Thus, the main ingredient in the solution is the following theorem
(that allows the construction of $\A_b$).  
\begin{theorem}\label{thm: main construction}
Let $\L$ be a library of RLC
components, each with $n_R$ return states, and let $\Aphi$ be a NWBA.
There exists an alternating B\"uchi automaton on trees (ABT) 
$\A$, 
with at most $O(|\Aphi|^2\cdot n_R)$ states, 
whose language is the set of composition trees for which there
exists a computation in the language of $\Aphi$.
\end{theorem}

Our main result follows from Theorem~\ref{thm: main construction}.
\begin{theorem}\label{thm: RLC synthesis solution}
The recursive library components realizability problem and 
the recursive library components synthesis problem are
2EXPTIME-complete.
\end{theorem}
\begin{proof}
The algorithm proceeds as follows. We first translate $\neg\varphi$
into an equivalent NWBA $\A_{\neg\varphi}$, with an exponential blow-up
\cite{AABEIL08}.  We then construct an ABT $\A$ for $\A_{\neg\varphi}$
according to Theorem~\ref{thm: main construction}, dualize $\A$ into an 
an alternating co-B\"uchi automaton on trees (ACT) $\A'$, and check 
$\A'$'s language for nonemptiness as in \cite{KV05c}.
If the specification is realizable, then the language of $\A'$
contains a regular composition tree, for which all computations satisfy
$\varphi$. Otherwise, the language of $\A'$ is empty. 
Given a regular composition tree $\tuple{[n_C]^*, \tau}$, it is induced 
by a transducer (without final states)
$T = \tuple{ [n_C], \L, Q, q_0, \delta, L}$, such that  
for every $w\in [n_C]^*$, we have $\tau(w)=L(\delta^*(w))$.
We assume, w.l.o.g. that the set $Q$ is the set $[|Q|]$ of natural
numbers, and that $q_0$ is the number 1.
A finite composition can now be constructed in the following way:
For every state $q\in Q$ there is a composition element 
$\tuple{q,C_q,f_q}$ in which $C_q = L(q)$, and for every $j\in [n_C]$
we have $f_i(j) = \delta(i,j)$.
It can then be shown that the constructed composition induces the same
infinite-state transducer as the regular composition tree (up to
component names) and therefore satisfies $\varphi$.

As for complexity, $\A$'s number of states is quadratic in $|\Aphi|$
and linear in $n$ and $b$ (upper bounding $n_R$ by $b$). 
(Note that quadratic in $|\Aphi|$ is exponential in $|\varphi|$).
The complementation of $\A$ into $\A'$ incurs no complexity cost. 
Finally, checking $\A'$ for emptiness is exponential in its number of
states. This provides a 2EXPTIME upper bound. 
For a lower bound, note that a ``goto'' can be seen as a call without a
return and LTL is a fragment of NWTL. Thus, a 2EXPTIME lower bound follows 
from the 2EXPTIME lower bound in \cite{LV09}.
\end{proof}

We now prove Theorem~\ref{thm: main construction}.
There are two sources of difficilty in the construction. First, we
have to handle here call-and-return computations in composition
trees. While computations in composition trees in \cite{LV09} always
go down the tree, computations here go up and down the tree. Second,
here we have to emulate NWBA on the computations of composition
trees, but we want to end up with standard tree automata, rather
then nested-word automata.

Intuitively, given a computation tree as input, our construction would
guess a computation of the input tree, in the language of $\Aphi$,
together with an accepting run of $\Aphi$, on the guessed computation.
As mentioned in in the discussion of Composition trees, however, a
computation of the composed system corresponds to a traversal in the
composition tree. Therefore, to guess the computation, i.e.,
the traversal in the input tree, and the computation of $\Aphi$ on it,
we employ 2-way-automata techniques.

Let $\Aphi=
\tuple{Q,Q_0,Q_f , P, P_0, P_f , \delta_c, \delta_i, \delta_r}$.
The construction of $\A$ is quite technical.
Below we present the construction of $\A$, where the
introduction of each part begins in an informal/intuitive discussion
and ends in a formal definition. 

\noindent
{\sf The states of $\A$}:
Intuitively, $\A$ reads an input tree $\tau$ and guesses an accepting run 
of $\Aphi$ on a computation of that input tree. The difficulty is that a 
computation cannot be guessed node by node, since when a computation 
enters a call node, we need to consider the return to that node.
Thus, when reading a node $v$ labeled by component $C$, the ABT $\A$
guesses an {\em augmented computation} of $C$ in which there are
{\em call transitions} from call states to re-entry states, and a
corresponding {\em augmented run} of $\Aphi$ (in which $\Aphi$'s state 
changes at the end of a call transition of $C$).
Of course, when $\A$ guesses a call transition it should also verify
that there exists a computation segment and a run segment of $\Aphi$, 
corresponding to that call transition.
To verify a call transition from $s_C^j$ to $s_R^k$, the ABT $\A$
sends a copy of itself, in an appropriate state, to $j$-child son of the
component being read.

In general, $\A$ has two types of states:  
states for verifying call transition (i.e. computations segments
between a call and its return), and states for
verifying the existence of computation suffixes that do not return. 
An example of a computation suffix that does not return is a
computation that follows a pending call.
States of the first type verify the feasibility of a computation
segment, and  there exists such a state every triple 
$\tuple{q,q',i}\in Q^2\times [n_R]$. 
If $\A$ reads a tree node $v$ in state $\tuple{q,q',i}$
it has to verify the existence of a computation in which a call was
made to $v$'s component when $\Aphi$ was in state $q$, and the first
return from $v$'s component is from the $i$-th return state $s_R^i$,
when $\Aphi$ is in state $q'$. 
States of the second type exist for every state $q\in Q$. 
If $\A$ reads a tree node~$v$ in state~$q$
it has to verify the existence of a computation suffix in
which a call was made to $v$'s component when $\Aphi$ was in state
$q$, and $\Aphi^q$ has an accepting run on that suffix. 
The initial state of $\A$ is of the second type: 
the initial state $q_0$ of $\Aphi$.

In fact, the state space of $\A$ must reflect one more complication. 
The ABT $\A$ not only has to guess a computation of a system and a run
of $\Aphi$ on it, the run of $\Aphi$ must be {\em accepting}.    
For that reason we also need to preserve information regarding 
$\Aphi$'s passing through an accepting state during a run segment.
In particular, when considering a call transition that stand for a
computation segment during which $\Aphi$ moved from $q$ to $q'$, it is
sometimes important whether during that run segment $\Aphi$ passed
through an accepting state. 
For that reason, states of the first type (that verify call
transitions) come in two flavors:
First, states $\tuple{q,q',i,0}$ that retain the meaning explained above. 
Second, states $\tuple{q,q',i,1}$ in which $\A$ has to verify that
in addition to the existence of a computation segment and an $\Aphi$ 
run segment as above, the run segment of $\Aphi$ {\em must} pass through
an accepting state.
Similarly, when $\A$ reads a component $C$ while in state $q$, it has to 
verify there is a computation that does not return on which $\Aphi^q$ has 
an accepting run. One of the ways this might happen, is that the
$C$ would make a pending call to some other component $C'$. If this is
the case, we need to keep track of whether an 
accepting state was seen from the entrance to $C$ until the call to
$C'$.
For that reason, states of the type $q$ also have two flavors:
$\tuple{q,0}$ and $\tuple{q,1}$ (where the second type stands for the 
constrained case in which an accepting state must be visited).
Thus, the formal definition of $\A$'s states set is 
$Q_\A = Q^2\times [n_R]  \times \{0,1\} \bigcup Q\times\{0,1\}$. 

\noindent
{\sf The transitions of $\A$}:
Intuitively, when $\A$ reads an input-tree node $v$ and its labeling
component $C$, the ABT $\A$ guesses an augmented computation and
a corresponding augmented run that take place in $C$.  Furthermore, 
for every call transition in the guessed augmented computation,
the ABT $\A$ sends a copy of itself to the direction of the call to
ensure the call transition corresponds to an actual computation
segment. Thus, if the call transition is from $s_C^j$ to $s_R^k$ and 
$\Aphi$ is moves from $q$ to $q'$ on that transition, then for some
$b\in \{0,1\}$ the ABT $\A$ sends a state $\tuple{q,q',k,b}$ to the
$j$-th direction (how $b$ is chosen is explained below).
The transition relation, therefore, has the following high level structure:
a disjunction over possible augmented computations and runs, where for
each augmented run a conjunction over all call transitions sending the
corresponding $\A$'s states to the correct directions. 

Before going into further detail, we introduce some notation:
Given an augmented computation of $C$ that begins in state $s$ and ends in
state $s'$ and an augmented run of $\Aphi$ on it that begins in state
$q$ and ends in state $q'$ we say that the beginning 
{\em configuration} is  $(s,q)$ and the final {\em configuration} is
$(s',q')$. Transitions of $\Aphi$ that have to do with calls or
returns have a hierarchical symbol associated with them. If the
composition $C$ is in state $s$, the ABT $\A$ is in state $q$ and
a hierarchical symbol $p$ is associated then the {\em configuration}
is $(s,q,p)$. Given two configuration $c_1$ and $c_2$ then $c_2$ is
{\em reachable in $C$} from $c_1$ if there exists computation segment
of $C$, that contain no call transitions, that begins in $c_1$ and
ends in $c_2$. 
The configuration $c_2$ is {\em reachable through accepting state in
$C$} from $c_1$ if  there exists computation segment of $C$, that
contain no call transitions, that begins in $c_1$ and ends in $c_2$,
and on which $\Aphi$ visits an accepting state. 

Next, we describe the transitions out of a state $\tuple{q,q',k,0}$. 
This is the simplest case as it does not involve analyzing whether an
accepting state of $\Aphi$ is visited. 
Assume $\A$ is in state $\tuple{q,q',k,0}$ when it reads a component $C$. 
Intuitively, this means that $\A$ has to guess an augmented
computation of $C$ that begins at $C$'s initial state, and ends in
$C$'s $k$-th return state, and an augmented run of $\Aphi$ on that
computation that begins in state $q$ and ends in state $q'$. 
In fact, instead of explicitly guessing the entire augmented
computation and run, what $\A$ actually guesses are only  
the call transitions appearing in the computation, and the state
transitions of $\Aphi$ corresponding to these call transitions.
These are needed as they define the states of $\A$ that will be sent
in the various directions down the tree.
The computation begins when $C$ is in its initial state $s_0$, and
$\Aphi$ is in state $q$. Thus the beginning configuration is $(s_0,q)$. 
The first call transition source is some call state $s_C^{j_1}$ of
$C$, some state $q_1$ of $\Aphi$ and a hierarchical symbol $p_1$ of $\Aphi$.   
Thus the first computation segment ends in configuration 
$(s_C^{j_1},q_1,p_1)$.  Note that it must be the case that the
configuration $(s_C^{j_1},q_1,p_1)$ is reachable in $C$ from $(s_0,q)$.
The target of the call transition is some configuration
$(s_R^{k_1},q_1',p_1)$. 
At this stage, i.e. when $\A$ reads $C$, the target configuration is
only constrained by sharing the hierarchical symbol with the call
transition source.
The constraints on the possible states in the target configurations
depend on components down the tree that $\A$ will read only at a later
stage of its run. 
The configuration which is the source of the next call transition,
however, again has to be reachable from $(s_R^{k_1},q_1',p_1)$. 

Our approach, therefore is to define a graph $G_C$ whose vertexes are 
configurations, and there exists an edge from a source configuration
to a target configuration if it is possible to reach the target from
the source (see earlier discussion of configurations).  Recall the notation
$C = \tuple{ \Sigma_I,\Sigma_O, S,s_0, s_e^R,  S_C, S_R, \delta, L}$, 
where $s_e^R = \{ s_e^i \}_{i=1}^{n_R}$,  $S_C = \{s_C^i\}_{i=1}^{n_C}$, 
and $S_R = \{ s_R^i\}_{i=1}^{n_R}$. 
The vertex set $V_C$ of $G_C$ is the union of four sets:
(1) Initial configurations $\{s_0\} \times Q$.
(2) Call configurations $S_C \times Q\times P$.
(3) Re-entry configurations $s_e^R \times Q\times P$.
(4) Final configurations $\{s_R^k\} \times Q$.

There are two types of edges in $G_C$. {\em Component edges} reflect
reachability in $C$. There is a component edge in $G_C$ from
configuration $c_1$ to configuration $c_2$ iff $c_2$ is reachable in
$C$ from $c_1$. {\em Call edges} capture call transitions and the
corresponding state changes in $\Aphi$. There is a call edge in $G_C$
between $c_1 = (s,q,p)$ and $c_2 = (s',q',p')$ if $s$ is a call state,
$s'$ is a re-entry state, and $p = p'$. 
   
An augmented computation and run of $\Aphi$ on it, correspond to a path in
$G_C$.  When $\A$ is in state $\tuple{q,q',k,0}$ and reads a component $C$
it guess a path in $G_C$ from $\tuple{s_0,q}$ to $\tuple{s_R^k,q'}$.  
If there exists such a path in $G_C$ there exists a short path of
length bounded by $|V_C|$, i.e. the number of vertexes in $G$. 
We denote by $\Path(q,q',s_R^k)$ the set of paths from $(s_0,q)$ to 
$(s_R^k, q' )$ of length bounded by $|V_C|$.
For each path $\pi\in \Path(q,q',s_R^k)$, we denote by
$E_C(\pi)$ the set of call edges appearing in $\pi$.
For a call edge $e = \tuple{(s_C^i,q,p),(s_e^j,q',p)}$, we denote 
$\dir(e) = i$, $\re(e) = j$, $q(e) = q$, and $q'(e) = q'$.
The transitions from $\tuple{q,q',k,0}$ are defined:

$$\delta(\tuple{q,q',k,0},C) = 
\bigvee_{\pi\in  \Path(q,q',s_R^k)}~\bigwedge_{ e\in E_C(\pi)}
(\dir(e), \tuple{ q(e),q'(e),\re(e),0}).$$

Intuitively, a path in $G_C$ is guessed and for each call edge $e$, 
the state $\tuple{ q(e),q'(e),\re(e),0}$ is sent in the direction of
the call, i.e. $\dir(e)$.

Next, we describe the transitions out of a state $\tuple{q,q',k,1}$. 
This case a very similar to the case of transitions out of
$\tuple{q,q',k,0}$ outlined above. 
The difference is that in this case $\Aphi$ must visit an accepting
state during its augmented run. 
There is no restriction, however, that the accepting state will be
visited when the control is held by the component $C$. It is possible
that the accepting state will be visited when some other (called)
component is in control. 
Intuitively, as in the $\tuple{q,q',k,0}$ case, the ABT $\A$ guesses a path
in $G_C$ from the initial to the final configuration, in addition,
$\A$ guesses an edge from the path in which an accepting state should
be visited. For component edges, it is possible to make sure that
guessed edges represent computations on which $\Aphi$ visits an
accepting state. For call edges, the task of verifying that an
accepting state is visited, is delegated to the state of $\A$ that is
sent in the direction of the call (by sending a state whose last bit
$b$ is 1).

Formally, a component edge in $G_C$ from configuration $c_1$ to
configuration $c_2$ is an {\em accepting edge} 
iff $c_2$ is reachable in $C$ through an accepting state from $c_1$. 
Note that if there exists a path from a configuration $c_1$ to
configuration $c_2$ that visits an accepting edge, then there exists
one of length at most $2|V_G|$ (a simple path to the accepting edge
and a simple path from it).
For $q,q'\in Q$, $s_R^k\in S_R$,  we denote by
$\Path_a(q,q',s_R^k)$ a set of pairs in which the 
first element is a path $\pi$ of length at most $2|V_C|$  from
$(s_0,q)$ to $( s_R^k, q' )$, and the second element is a function
$f$ mapping the edges in $\pi$ into  $\{0,1\}$ such that:
\begin{compactenum}
\item Exactly one edge is mapped to 1, and
\item If the edge mapped to 1 is a component edge then it is also an
      accepting edge. 
\end{compactenum}
Finally,
$$\delta(\tuple{q,q',k,1},C) = 
\bigvee_{(\pi,f)\in
  \Path_a(q,q',s_R^k)}~\bigwedge_{e \in E_C(\pi)}
(\dir(e), \tuple{ q(e),q'(e),\re(e),f(e)}).$$

Next, we describe the transitions out of a state $\tuple{q,b}$, for 
$b\in\{0,1\}$, in which 
$\A$ has to verify there exists an accepting
augmented computation of $C$ that does not return, and a run of
$\Aphi^q$ on it. 
There are three distinct forms such a computation might take. (1) First, it
is possible that the computation has a infinite suffix in which $C$
remains in control.  
(2) Second, it is possible that the eventually the component
makes some pending call.
(3) Finally, it is possible that the computation contains infinitely
many calls to, and returns from, other components.
We deal with each of the case separately, we construct a partial
transition relation for each case, the transition relation itself is
the disjunction of these three parts.    

First, to deal with infinite (suffixes) of computations that never
leave the component, we modify the graph $G_C$ to consider such runs. 
We introduce a new vertex $\bot$ that intuitively stand for ``an
infinite (suffix) of a computation in $C$, and an accepting run of
$\Aphi$ on it''. 
There is an edge from a configuration $c$ to $\bot$,
if there is an exists an infinite computation of $C$ that begins in
configuration $c$, never enters an exit state, and there exists an 
accepting run of $\Aphi$ on it. There are no edges from $\bot$.

The first part of the transition relation is
$$\delta_1(\tuple{q,b},C)  = \bigvee_{\pi\in  Path(q,\bot)}~\bigwedge_{ e \in E_C(\pi)}
      (\dir(e), \tuple{ q(e),q'(e),\re(e),0})$$

Second, we have to deal with computation segments that end in a pending 
call.  These types of computations are easily dealt with in terms of
paths in $G_C$ to a configuration in which the state is a call state. 
We would like to note two details. 
First, note that by the definition of an accepting run of an NWBA, 
the hierarchical symbols associated with pending calls must be from
the set $P_f$.
Second, note the difference between states of type $\tuple{q,0}$ and
type $\tuple{q,1}$. In the $\tuple{q,0}$ case there is no constraint that 
has to do with $\Aphi$'s accepting states. Therefore, the second part of
the transition relation is  
$$\delta_2(\tuple{q,0},C) = 
\bigvee\limits_{
\begin{array}{c} s_C^k\in S_C,\\ q'\in Q,\\ p\in P_f\end{array}}
~ \bigvee_{\pi\in Path(s_c^k,q,q',p)}
~ \bigvee_{b\in \{ 0, 1 \}} ((k,\tuple{ q,b }) \land
      \bigwedge_{ e \in E_C(\pi)} (\dir(e), \tuple{ q(e),q'(e),\re(e),0}))$$
In the $\tuple{q,1}$, case an accepting state of $\Aphi$ must be
visited, therefore the second part of the transition relation is
$$\delta_2(\tuple{q,1},C) = 
\bigvee\limits_{
\begin{array}{l}s_C^k\in S_C,\\ q'\in Q, \\p \in P_f\end{array}} 
~\bigvee_{(\pi,f)\in Path_a(s_c^k,q,q',p)}
 ~\bigvee_{b\in \{ 0, 1 \}} ((k,\tuple{ q,b}) \land
      \bigwedge_{ e \in E_C(\pi)} (\dir(e), \tuple{ q(e),q'(e),\re(e),f(e)}))$$

We have to deal with suffixes of computation that contain
infinitely many call to, and return from, other components. 
Such computations must contain a configuration that appears twice.  
A {\em $\rho$-path} in $G_C$ is a path in $G_C$ in which the last vertex
is visited more then once along the path (intuitively, closing a cycle).
The part of the path between the first and last
occurrences of the last vertex is the {\em cycle}.
As we require $\Aphi$'s run to accept, an
accepting state from $Q_f$ should be visited during a segment of a
computation that correspond to an edge on the cycle.
An {\em accepting $\rho$-path} is a path in which one of the edges  
along the cycle is accepting. 
There exists an accepting $\rho$-path iff there exists an accepting
$\rho$-path of length at most $3|V_C|$ 
(a simple path to the cycle, and a cycle of length at most $2|V_C|$). 

For $q,\in Q$ we denote by $\rhopath(q)$ a set of
pairs in which: (1) the first element $\pi$ is a $\rho$-path
of length at most $3V_C$ starting at $(s_0,q)$; (2) the second element
is a function $f$  mapping the edges in $\pi$ into $\{0,1\}$ such that: 
(1) exactly one edge  is mapped to 1, this edge is on the cycle, and
(2) if the edge mapped to 1 is a component edge then it is also an
accepting edge. 
The third part of the transition relation is 
$$\delta_3(\tuple{q,b},C) =  
\bigvee_{(\pi,f)\in \rhopath(q)}~
  \bigwedge_{e \in E_C(\pi)}(\dir(e), \tuple{ q(e),q'(e),\re(e),f(e)}) $$
Finally, for a state $\tuple{q,b}$ the transition relation is 
$$\delta(\tuple{q,b},C) = \delta_1(\tuple{q,b},C) \lor 
\delta_2(\tuple{q,b},C) \lor \delta_3(\tuple{q,b},C)$$ 
This concludes the definition of the transition relation

\noindent
{\sf Accepting states of $\A$}:
Finally, the set $F$ of $\A$'s accepting states is the set 
$Q \times \{ 1 \}$. Intuitively, in an accepting run
tree of $\A$, each path is either finite, i.e. ends a nodes whose
transition relation is \true, or an infinite path of states that
correspond to pending calls. 
For the run to be accepting, an accepting state must be visited
infinitely often along such infinite path of pending calls.  
As we defined the accepting-states set to be $Q \times \{ 1 \}$,
an infinite path of pending calls is accepted iff in the run of
$\Aphi$ visits an  $\Aphi$ accepting state infinitely often. 
This concludes the main construction, 

We now prove the correctness in several stages. 
First, we prove a claim regarding states of the form $\tuple{q,q',i,b}$.
\begin{claim}
\label{thm: return states claim}
For a composition tree $T$, there exists a finite accepting run tree
of $A^{\tuple{q,q',i,b}}$ on $T$ iff there exits a computation $\pi$
of the composition induced by $T$, such that:
\begin{compactenum}
\item 
  $\pi$ ends by returning from the $i$-th return state $s_R^i$ of
  $T$'s root.
\item 
  there exists a run $r$ of $\Aphi^q$ on the word induced
  by $\pi$ that ends in $q'$.  
\end{compactenum}
Furthermore, for states $\tuple{q,q',i,1}$ the iff statement is true
for a run $r$ that visits an accepting state from $Q_f$.  
\end{claim} 
\begin{proof}
Assume first that there exist computation $\pi$ and run $r$ as
claimed. We prove that there exists a finite accepting run tree of
$\A^{\tuple{q,q',i,b}}$ on $T$.
As the computation $\pi$ returns from the root, the depth $h$ of the
subtree traversed by $\pi$ in $T$ is bounded.    
The proof is by induction on the depth $h$.
The base case is a depth 1, i.e., only the root component is
traversed. Then, the existence $\pi$ implies there exists an edge in
$G_C$ from $\tuple{s_0,q}$ to $\tuple{s_R^i,q'}$. Therefore, 
there exists a path in $G_C$, between these vertexes, that does not
contain any call edges. Thus, the transition relation evaluates to
$\true$, implying that there exists a finite accepting run of
$\A^{\tuple{q,q',i,0}}$ on $T$.
Furthermore, if $r$ visits a state from $Q_f$ then the relevant edge
is an accepting edge and there is an accepting run of
$\A^{\tuple{q,q',i,1}}$ on $T$.  
Assume now, the induction hypothesis for traversal of maximal depth
$h$, we prove it for traversals of maximal depth $h+1$.
The computation $\pi$ can be broken into segments in which the control
is in the root component and segments in which some other (called)
components are in control. Each segment corresponds to an edge of $G_C$,
where segments of computation in which the root is in control,
correspond to component edges, and the rest of the segments correspond
to call edges. 
Each call edge, correspond to a successor of the root in the
composition tree, and for each call edge, the induction hypothesis
implies the existence of accepting run tree on the corresponding
composition subtree. Thus there exists an accepting run tree as claimed.   
Furthermore, if $r$ visits $Q_f$ then the visit is made
during some computation segment. The edge corresponding to that
computation segment can be mapped to $1$ by the function $f$ from the
definition of the transition relation for $\tuple{q,q',i,1}$. 
It follows that if $r$ visits a state from $Q_f$ then there exists a
an accepting run of $\A^{\tuple{q,q',i,1}}$ on $T$. 

Assume now a finite accepting run tree of $\A^{\tuple{q,q',i,b}}$
exists, we prove the existence of a computation $\pi$ and a run $r$ as
needed.  The proof is by induction on the height $h$ of the accepting
run tree.  The base case is a run tree of height 1.
Then, the transition relation $\delta$ must evaluate to $\true$ on the
root. Thus, the path in $G_C$ contains no call edges, and therefore
by the definition of $\delta$ there exist $\pi$, and $r$ as claimed. 
Furthermore, if $b=1$, the component edge must be an accepting
edge implying that $r$ visits $Q_f$. 
Assume now, the induction hypothesis for run trees of height $h$, we
prove it for run trees of height $h+1$.
The run-tree root is labeled by some set $S$ of pairs of directions
and $\A$-states  that satisfy $\delta$.  
This choice of states and directions corresponds to a path
in $G_C$, in which some edges are call edges and some are component
edges. By the definition of $\delta$ there exist computation segments
corresponding to component edges, and by the induction hypothesis
there exist computation segments corresponding to call edges. Splicing
these computation segments together we get the a computation $\pi$,
and $r$ as claimed. 
Furthermore, if $b=1$ then one of the edges is an accepting edge and
therefore, $r$ visits $Q_f$.
\end{proof}

By very similar reasoning, we can show that there exists a finite
accepting run tree of $A^{\tuple{q,b}}$ on a composition tree
$T$, iff there exists a computation $\pi$ of $T$ such that: 
(1) $\pi$'s traversal is bounded in a finite subtree of the
composition tree ; (2) $\pi$ never returns from the root of $T$; and
(3) there exists an accepting run $r$ of $\Aphi$ on $\pi$.
Unlike, the $\tuple{q,q',i,b}$, however, we have also to consider runs
that are not bounded in a finite subtree of $T$.
Next, we show that it is enough to consider computations that
make infinitely many pending calls.

\begin{observation}\label{thm: bounded stack observation}
For a library $L$, an NWBA $\Aphi$ and a composition tree $T$ 
if there exists a computation $\pi$ of $T$, in $L(\Aphi)$, in which a
node $v\in T$ is visited infinitely often then there exists 
computation $\pi'$ of\ $T$, in $L(\Aphi)$, that only traverses a
finite subtree of $T$.  
\end{observation}
\begin{proof}
First, note that it is enough to show that there exists a computation
$\pi'$ of $T$, in $L(\Aphi)$, such that $\pi'$ only traverses a finite
subtree of the subtree rooted at $v$ (regardless of what happens
outside that subtree). 
The reason is w.l.o.g. $v$ can be assumed to be a node of minimal
depth that is visited infinitely often by $\pi$. 
As such, the computation must eventually remain in the subtree
rooted at $v$ (since if $v$ is not the root, $v$'s predecessor is
visited only finitely often).

Next, let $\pi_1,\pi_2$ be two computation segments of $\pi$, and
$r_1,r_2$ be the corresponding parts of $\Aphi$'s accepting run on
$\pi$ such that:
\begin{compactenum}
\item
$\pi_1$ and $\pi_2$ begin by entering the same call state $s_C^i$.
\item 
$r_1$ and $r_2$ begin by the same $\Aphi$ state $q$.
\item 
$\pi_1$ and $\pi_2$ end when the control is returned to $v$ by the
same re-entry state $s_e^j$.
\item 
$r_1$ and $r_2$ end in the same $\Aphi$ state $q'$.
\item 
$r_1$ visits $Q_f$ iff $r_2$ visits $Q_f$.
\end{compactenum}
Then, $\pi_1$ and $\pi_2$ are interchangeable while the
resulting computation still satisfies $\varphi$. 
Thus, while $v$ is returned to infinitely often, there are only
finitely many equivalence class of interchangeable computation
segments. Choosing a single representative from each equivalence
class, we  can splice a computation whose traversal depth is
bounded by the traversal depths of the representatives.  
\end{proof}

Observation~\ref{thm: bounded stack observation} implies 
that if there is a computation, in $L(\Aphi)$, that traverses an
unbounded subtree of the composition tree, and does not make
infinitely many pending calls, then there is also 
a computation, in $L(\Aphi)$, that traverses a finite subtree of the
composition tree.
Therefore, when considering computations that traverse an unbounded depth
subtree of a composition tree, it is enough to consider compositions
in which the computation, whose word is in $L(\Aphi)$, has infinitely
many pending calls. 
The definition of $\A$'s accepting states set ensures correctness with
respect to such computations. 
An accepting run tree, of $\Aphi$ on $T$, with an infinite path, must
visit infinitely often an accepting state (i.e., a state
$\tuple{q,1}$) which means it is possible to construct a
computation of $T$ that makes infinitely many pending calls, and on
which $\Aphi$ would have an accepting run.
On the other hand, an accepting computation of $\Aphi$ that makes
infinitely many pending calls, implies the existence of an accepting
run tree of $\A$, with an infinite path that visits an accepting state
infinitely often.

Finally we provide a complexity analysis. 
For a NWBA $\Aphi$, with $n_{\Aphi}$ states, and a library $\L$ with
$m_L$ components in which the components are of size $m_C$,  
the construction presented here, creates an ABT $\A$ with at most
$O(n_{\Aphi}^2\cdot m_C)$ states. Note, however, that the number of states
does not tell the entire story. 
First, the computation of $\delta$ involves
reachability analysis of the components. Luckily, the reachability
analysis is done separately on each component (in fact, the Cartesian
product of each component with $\Aphi$) and therefore the complexity
is $O(n_{\Aphi}\cdot m_C \cdot m_L)$. 
On the other hand, $\A$ is an alternating automaton with a transition
relation that may be exponential in the size of the its state space. 
Thus, $\A$ can {\em not} be computed in space polynomial in the parameters.
The computation of $\A$ involves an analysis of the paths in $G_C$ and
requires space polynomial in $n_{\Aphi}$ and $m_C$.

\section{Discussion}
We defined the problem of NWTL synthesis from library of recursive
components, solved it, and shown it to be 2EXPTIME-complete. We now note
that the ideas presented above are quite robust with respect to possible
variants of the basic problem.

The model was chosen for simplicity rather than
expressiveness, and can be extended and generalized.
First, we can consider several call values per component. 
This translates to each component having a set $S_0\subseteq S$ of
initial states (rather than a single initial state $s_0\in S$). 
Next, we can add greater flexibility with respect to return values.
A single return value may have different meanings on different calls. 
Therefore, compositions might be allowed to perform some ``return-value
translation''; matching return states to re-entry states per call, 
rather than matching return states to re-entry states uniformly. 
This can be modeled by augmenting each composition element
$\tuple{i,C_i,f_i}$ by another function $r_i:S_C\to ( [n_R] \to s_e^R )$
that maps each call state into a matching of return values to re-entry 
states.  The synthesis algorithm, for the augmented model, remains 
almost the same.  In the augmented model, a component implementation 
depends on the call value $s_0\in S_0$ and the $r_i$ function.
Therefore, instead of working with composition trees, labeleded by $\L$,
we'd work with augmented composition trees, labeled by tuples
$\tuple{C,s_0,r_i}$.  Our algorithm and analysis can then be
extended appropriately.

Another possible extension might be to consider bounded call stacks.
Theoretically, ``call and return'' models allow for unbounded call 
stacks.  Real life systems, however, have bounded call stacks.
One can consider a variant of the synthesis problem, in which the
output must have bounded call stack, where the bound is
an output of the synthesis algorithm, rather then an apriori given 
input.  To adapt the algorithm to this case, we have to find a finite
composition tree in which all computations satisfy $\varphi$, as well
as no computation makes a call from a leaf (ensuring bounded stack). 
To that end, we construct two alternating automata on finite trees 
(AFTs).  First, an AFT $\A_1$ for finite composition trees in which 
there exists a computation violating $\varphi$. The AFT $\A_1$ is 
simply the ABT from Theorem~\ref{thm: main construction}, when 
considered as an AFT, and in which no state is considered accepting. 
In addition, we construct an AFT $\A_2$ that accepts trees that may 
perform a call from one of the leaves (see longer version of this
paper.) The union of the languages of $\A_1$ and $\A_2$ contain all finite
composition trees that do not realize $\varphi$.  An AFT for the union 
can then be complemented and checked for emptiness as in the infinite 
case.  Thus, the solution techniques presented in this paper are quite
robust and extend to natural variants of the basic model.

\noindent
{\sf Acknowledgements}
Work supported in part by NSF grants CCF-0728882, and CNS 1049862, 
by BSF grant 9800096, and by gift from Intel.

\bibliographystyle{eptcs}
\bibliography{reactive}
\end{document}